\documentclass[a4paper,oneside,11pt]{article}
\usepackage{amsmath,amssymb, amssymb, amsfonts, setspace, latexsym, epsfig, epsf, rotating, longtable, setspace, a4wide,verbatim, caption}
\usepackage{amsthm}
\usepackage[utf8x]{inputenc}
\theoremstyle{plain} \newtheorem{theorem}{Theorem}
\theoremstyle{plain} \newtheorem{cor}{Corollary}%
\theoremstyle{definition} \newtheorem{rem}{Remark}%

\newcommand{\R}{\ensuremath{\mathbb{R}}}%
\newcommand{\FJ}{\ensuremath{\mathcal{J}_0^{-1}}}%
\newcommand{\FJI}{\ensuremath{\mathcal{J}_0}}%
\newcommand{\FI}{\ensuremath{\mathcal{I}_0^{-1}}}%
\newcommand{\FIjoi}{\ensuremath{\mathcal{I}_{Joint}^{-1}}}%
\newcommand{\FIjoor}{\ensuremath{\mathcal{I}_{Joint}}}%
\newcommand{\FIcom}{\ensuremath{\mathcal{I}_{comp}^{-1}}}%

\newcommand{\Lower}[1]{\smash{\lower 1.5ex \hbox{#1}}}

\usepackage{natbib}
\newcommand{\ct}[1]{(\cite{#1})}
\newcommand{\ctp}[1]{\cite{#1}}
\newcommand{\ctt}{\citet}

\begin{document}

\begin{center}
{\large{\bf
Evaluating HWE and Association in Genome Wide Association Studies: A Unified Procedure 
}}
\end{center}

\begin{center}
 Stefan Boehringer$^{1*}$ and Hajo Holzmann$^2$
\end{center}


\noindent$^1$  	 Dept of Biomedical Data Sciences,\\ 	
	Leiden University Medical Centre,\\ 
	2300 RC Leiden, The Netherlands,\\ Phone: 	 +31 71 5269743  \\
Email: correspondence@s-boehringer.org


\vskip 5mm

\noindent$^2$  Fachbereich Mathematik und Informatik, Philipps-Universität Marburg, Germany\\
Hans-Meerweinstr., 35032 Marburg, Germany\\
Email: holzmann@mathematik.uni-marburg.de\\

\noindent$^*$Corresponding author
\vskip 5mm\noindent


{\bf Running title:} Unifying HWE and Association Testing

\vskip 5mm\noindent


\newpage
\doublespacing
{\bf\Large Summary}\par

In genome wide association studies (GWASs) based on a case-control design, single nucleotide polymorphisms (SNPs) are typically evaluated for an association test and a Hardy-Weinberg equilibrium (HWE) goodness-of-fit test. SNPs are then excluded from analysis based on a HWE cutoff to avoid false positives. 
In order to avoid cutoffs based on arbitrary threshold values, we propose a conditional genotype--based test that conditions the Pearson $\chi^2$-statistic in the 3x2 contingency table on the $\chi^2$-statistic for HWE in the control group, and develop the relevant asymptotic distribution theory. We show by simulations that our test in most scenarios is more powerful than two competing retrospective procedures.
Another important advantage of the proposed method is a better ranking of SNPs in GWASs as HWE is accounted for in computing p-values of SNP association. We demonstrate this effect on a data set in an alopecia study.
In conclusion, our test makes separate HWE testing superfluous by providing a unified framework and strictly improves on the standard procedure in terms of power and interpretability, thereby making replication more cost effective and improving subsequent fine mapping.\par

\newpage

\section{Introduction}

Genome wide associations studies (GWASs) are currently used in epidemiological studies to identify genetic variants with predictive power for a phenotype \ct{kruglyak_2008}. In such a study typically 500,000-1,000,000 single nucleotide polymorphisms (SNPs) are scrutinized, each being tested for significant association with a phenotype.

We focus on case-control studies for which several testing strategies exist which we briefly review. 
First, generalized linear models (GLMs), which regress the phenotype on observed genotypes, lead to a class of tests that includes the popular Armitage trend test. 
Second, retrospective models are used, where the genotype distribution is compared conditional on the group status. This class includes the Pearson $\chi^2$-test that compares genotype frequencies across groups as well as the approach in \ct{chen_2007}, where frequencies in cases are compared with a projection of control frequencies.
Third, heuristic combination strategies of GLMs with tests that make use of information about the genotype distribution such as Hardy-Weinberg equilibrium (HWE) deviations have been suggested ({\it e.g.} \ctt{wang_2008,song_2006}).\\
In this paper, we propose a novel solution for analyzing HWE, an aspect of case-control studies that has not been sufficiently addressed so far. 
HWE describes the assumption that the two observed copies (alleles) at a genetic locus are randomly and independently drawn from the previous generation.
We propose a conditional genotype--based test that conditions the Pearson $\chi^2$-statistic in the 3x2 contingency table on the $\chi^2$-statistic for HWE in the control group.
A motivation for our approach stems from the fact that HWE is usually assumed to hold in the control group as has been empirically confirmed in many studies ({\it e.g.} \ctp{yeager_2007}). The HWE assumption in controls can be tested by means of a goodness-of-fit test, and deviations from HWE often indicate systematic errors, {\it e.g.} in the process of observing genotypes (genotyping errors), sample contamination or confounding by population strata (see below). It is therefore desirable to put more weight on SNPs for which the control group closely follows HWE but otherwise to reduce the weight. Below we shall show that conditioning on HWE in the control group precisely has this effect.

In the GWAS situation, SNPs departing from HWE in controls are excluded from further analysis based on testing with a small significance level ({\it e.g.} \ctp{gudmundsson_2007,yeager_2007,salmela_2008}). However, such a strategy is problematic. 
The significance level for the HWE test used to exclude a SNP is chosen rather arbitrarily and has indeed shifted historically due to the increasing number of SNPs used in GWASs and the related multiple testing problem ({\it e.g.} compare \ctp{gudmundsson_2007,salmela_2008}). 
At level $\alpha = .05$ for the HWE test in a study with 500,000 SNPs, 25,000 SNPs are expected to be excluded even if HWE holds true. This situation has led to suggestions of choosing $\alpha = 0.001$ or lower thereby leaving SNPs with moderate deviation from HWE in the study. Our condional approach avoids this arbitrary thresholding and instead offers a continuous correction for deviations from HWE.



In the literature, HWE was identified as an important criterion both to evaluate quality of association studies \ct{salanti_2005} as well as to detect systematic errors \ct{xu_2002}. The latter reference notes the correlation between false positives and HWE deviations, which is attributed mainly to genotyping errors. However, this correlation is naturally implied by our results in the absence of systematic error. \ctt{wittke-thompson_2005} develop a goodness-of-fit test for HWE by separating deviation due to effects of some genetic model and ``genuine'' HWE deviations. Strong assumptions about the genetic model are needed. \ctt{wang_2008} follow a similar idea by assuming that HWE deviations in cases may be due to a genetic effect. They focus on global testing by adopting a so called tail-strength measure following \ctt{taylor_2006}. \ctt{song_2006} consider a weighted average between the Armitage trend test and a trend statistic comparing HWE in cases and controls. The weighting factor is chosen arbitrarily and is justified by 
simulations. The most closely related work proposes a test, called EHWE (expected HWE) test in the following, which ignores possible deviations from HWE in controls altogether \ct{chen_2007}. This approach is later shown to often have unfavorable properties in the GWAS setting. 
Similar to the EHWE test our procedure is more powerful than an ordinary genotype based test if controls follow HWE but in contrast to the EHWE it offers robustness against deviations from HWE, a feature of particular importance in GWASs.\par


The paper is organized as follows: In section \ref{sec_theory} we introduce a new parametrization of the testing problem and develop the asymptotic theory for a conditional test. In Section \ref{sec_simulation}, we scrutinize finite-sample properties of the test by means of simulation studies and compare the test with alternative approaches. Section \ref{sec_data} is concerned with the analysis of a data set on an alopecia phenotype. This analysis gives some insight into behavior under mis-specification as some mild deviation from HWE in controls is observed. In Section \ref{sec_discussion} we conclude with a discussion in which we summarize the properties of our proposed procedure and highlight the benefits for practical data analysis. Proofs and further technicalities are contained in a supplementary appendix. 

\section{Asymptotic conditional distribution}\label{sec_theory}


\subsection{Reparametrization of the multinomial distribution}\label{sec:repar}

The genotype distribution for $N$ individuals at a SNP can be parametrized by a multinomial distribution $M(N;\pi_1,\pi_2,\pi_3), \pi_1+\pi_2+\pi_3=1,$ with three possible resulting categories and $N$ repetitions. Categories represent genotypes determined by allele counts of one arbitrarily chosen allele.\par
In order to motivate our new parametrization, we recall the most simple model of a neutral coalescent process for a fixed population size \ct{hudson_1990}. In this model each allele is drawn randomly and independently from the previous generation with replacement. The grouping of two alleles into a genotype (an individual) is not needed in the model as it seeks to describe allele behaviors in terms of frequency changes and  time ({\it e.g.} number of generations) that has passed since two alleles were drawn from the same ancestor.
In this setting, the allele frequency $\rho$ suffices to describe the distribution per generation. Following this notion, we interpret the genotype distribution as the result of a two stage process. First, alleles are selected at random for transmission to the current generation (gamete formation), a process that only depends on the allele frequency $\rho$ in the previous generation. Second, alleles are paired at random to form genotypes (individuals) in the current generation. Then the genotype distribution follows the HWE expectation if $\pi_j = f_j(\rho), j = 1, 2, 3$, where
$$
	\big(f_1(\rho), f_2(\rho), f_3(\rho) \big) = \big(\rho^2, 2 \rho (1-\rho), (1 - \rho)^2\big).
$$ 
 We introduce a second parameter $\eta$ that measures deviation of genotype formation from a HWE expectation. We define $(\rho,\eta) = T(\pi_1,\pi_2)$, where
\begin{eqnarray*}
	\rho & = & \pi_1 + \pi_2/2\\
	\eta & = & \text{sgn}\big(\pi_2 -f_2(\rho)\big)
		\Big(\sum_{j=1}^3 \frac{(\pi_j - f_j(\rho))^2}{f_j(\rho)} \Big)^{1/2}.
\end{eqnarray*}
One can recover the original parameters of the multinomial by letting $ 	\pi_{1} = \rho - \pi_{2}/2$ and $\pi_{2} = f_2(\rho) (1 + \eta)$ thereby retaining the full information on the genotype distribution (see supplementary appendix Section \ref{asec:parnew}). The parameter $\eta$ measures an excess or deficit of heterozygous genotypes (allele count of 1 in the genotype) relative to HWE. Not only does this parametrization allow to formulate our problem more conveniently, it also has some merits in reformulating existing statistical procedures concerning HWE. For example, we will later discuss the EHWE test which has a very simple null hypothesis in our parametrization \ct{chen_2007}. We also note, that graphical representations of genotype frequencies might be more easily interpreted as compared to De Finetti diagrams \ct{edwards_2000}.

\subsection{Maximum likelihood estimates and Fisher information of the reparametrized likelihood}

Consider again the multinomial distribution $M(N;\pi_1,\pi_2,\pi_3)$ as above. Let $X = (n_1,n_2,n_3) \sim M(N;\pi_1,\pi_2,\pi_3)$. The maximum likelihood estimates for $\pi_i$, $i=1,2,3$, are given by
$ \hat \pi_i = n_i/N$.
We have
\[ \sqrt{N} \big((\hat \pi_1,\hat \pi_2)^T - (\pi_1, \pi_2)^T \big) \stackrel{d}{\to} N(0,\FJ),\]
where 
\[ \FJ = \left(\begin{array}{cc} \pi_1 (1-\pi_1) & - \pi_1 \pi_2 \\ - \pi_1 \pi_2 & \pi_2 (1-\pi_2) \end{array}\right),\qquad \FJI = \left(\begin{array}{cc} \frac{1}{\pi_1} +\frac{1}{\pi_3} & \frac{1}{\pi_3} \\ \frac{1}{\pi_3} & \frac{1}{\pi_2} +\frac{1}{\pi_3} \end{array}\right),\]
$\FJI$ being the Fisher information of the multinomial.\\*[0.2cm] 
The ML estimates of $(\rho, \eta)$ are then given by $(\hat \rho, \hat \eta) = T(\hat \pi_1,\hat \pi_2 )$. Their asymptotic distribution is easily derived from the $\delta-$method. Indeed, 
\[ \sqrt{N} \big((\hat \rho,\hat \eta)^T - (\rho, \eta)^T \big) \stackrel{d}{\to} N(0,\FI),\]
where $ \FI =   DT\, \FJ \, DT^T$ and $DT$ denotes the Jacobian of the map $T(\pi_1, \pi_2)$ defined above. Straightforward algebra shows that
%
%
%
\begin{eqnarray}
\FI	& = & \frac{1}{2} \left(
\begin{array}{ll}
	(1 -\eta) R  &  \eta  E (1 - 2 \rho) \\
	\eta  E (1 - 2 \rho) & - E V/R
\end{array}\right),  \qquad \begin{array}{ccl}
	R & = & \rho (1 - \rho), \qquad E = 1 + \eta \\
	V & = & \eta^2(1 -2 \rho )^2 -2 R + \eta  (2 - 6 R)\end{array}
\label{eq:fisherinfotranspar}
\end{eqnarray}
Observe that if HWE holds, $\eta =0$ and thus $\FI = \text{diag}\, (R/2, 1)$ is a diagonal matrix and the estimators $(\hat \rho, \hat \eta)$ are asymptotically independent.

\subsection{Conditional tests}\label{sec:comparision}
Consider two multinomial distributions $M(N_i;\pi_{i,1},\pi_{i,2},\pi_{i,3})$, $i=1,2$, corresponding to control group for $i=1$ with $N_1$ repetitions and to the case group for $i=2$ with  $N_2$ repetitions. Let $X_i = (n_{i,1},n_{i,2},n_{i,3}) \sim M(N_i;\pi_{i,1},\pi_{i,2},\pi_{i,3})$.  The hypothesis of homogeneity is
\[ H : \pi_{1,j} = \pi_{2,j} =: \pi_j,\qquad j=1,2,3.\]
We set $n_{\cdot j} = n_{1,j} + n_{2,j}$ and $N = N_1 + N_2$.
The unrestricted log-likelihood is given by 
\[  \mathcal{L}(\pi_{1,1},\pi_{1,2},\pi_{2,1},\pi_{2,2}) = \sum_{i=1,2} \sum_{j=1,2} n_{i,j} \log \pi_{i,j} + \sum_{i=1,2} n_{i,3} \log \pi_{i,3},  \]
where $\pi_{i,3} = 1-\pi_{i,1} - \pi_{i,2} $, and the restricted log-likelihood under $H$
\[ \mathcal{L}(\pi_{1},\pi_{2}) = \sum_{j=1,2} n_{\cdot j} \log \pi_{j} + n_{\cdot 3} \log \pi_{3}.  \]
The unrestricted ML estimates are $\hat \pi_{i,j} = n_{i,j}/N_i$, $i=1,2$, $j=1,2$, and in the equivalent parametrization, $(\hat \rho_i,\hat \eta_i) = T(\hat \pi_{i,1}, \hat \pi_{i,2})$. Further, the restricted ML estimates under $H$ are given by $\hat \pi_j = n_{\cdot j}/N$. Let $\mathcal{D}(Z)$ denote the distribution of the random variable $Z$, and let $d$ be a metric on the probabilities on $\R$ which metrizes weak convergence. 
\begin{theorem}\label{the:mainresultgen}
Suppose that $H$ holds. If $c=N_1/N$ is bounded away from $0$ and $1$, then for the asymptotic conditional distribution of the LRT statistic, given $\hat \eta_1$, we have as $N \to \infty$ that 
\begin{equation}\label{eq:asymptoticgen}
	d\Big(\mathcal{D}\Big(
		2 \big(\mathcal{L}(\hat \pi_{1,1},\hat \pi_{1,2},\hat \pi_{2,1},\hat \pi_{2,2})
			- \mathcal{L}(\hat \pi_{1},\hat \pi_{2}) \big)| \hat \eta_1
		\Big), \mathcal{D}\big(W| \hat \eta_1\Big)\Big) = o_P(1),
\end{equation}
where $\mathcal{D}\Big(\cdot| \hat \eta_1\Big)$ is the conditional distribution given $\hat \eta_1$, 
\begin{equation}\label{eq:asympconddistr}
 W  =  \big( \Sigma^{1/2}\, Z + \mu(\hat \eta_1) \big)^T  \FIjoor   \big(\Sigma^{1/2} Z + \mu(\hat \eta_1) \big),
\end{equation}
$Z \sim N(0,I_4)$ is independent of $\hat \eta_1$, 
\begin{equation}\label{eq:noncentralparcond}
	\mu(\hat \eta_1) = 	\sqrt{N} \big(\hat \eta_1 - \eta \big)\left(
	\begin{array}{c}
		(1 - c) \eta  R  (2 \rho -1)/V \\
		c		\eta  R  (2 \rho -1)/V \\
	1-c \\
	-c	
	\end{array}
	\right).
\end{equation}
and the matrix $\Sigma$ is given in (\ref{eq:covasymcondright}) in the appendix.
\end{theorem}
\begin{cor}\label{the:mainresult}
Suppose that $H$ holds and that further the control group is in HWE, i.e. $\eta_1 = 0$. If $c=N_1/N$ is bounded away from $0$ and $1$, then for the asymptotic conditional distribution of the LRT statistic, given $\hat \eta_1$, we have as $N \to \infty$ that 
\begin{equation}\label{eq:asymptotic}
	d\Big(\mathcal{D}\Big(
		2 \big(\mathcal{L}(\hat \pi_{1,1},\hat \pi_{1,2},\hat \pi_{2,1},\hat \pi_{2,2})
			- \mathcal{L}(\hat \pi_{1},\hat \pi_{2}) \big)| \hat \eta_1
		\Big), \mathcal{D}\big((Z_1 + Z_2) | \hat \eta_1\Big)\Big) = o_P(1),
\end{equation}
where $\mathcal{D}\Big(\cdot| \hat \eta_1\Big)$ is the conditional distribution given $\hat \eta_1$, the random variable $Z_1$ and the random vector $(Z_2,\hat \eta_1)$ are independent, $Z_1 \sim \chi^2_1$ and $Z_2|\hat \eta_1 \sim c \chi^2_1\big(N \hat \eta_1^2 (1-c)\big)$. 
\end{cor}

Theorem \ref{the:mainresult} as well as Corollary \ref{the:mainresult} are proven in the supplementary appendix Section \ref{app:proof}. 

Next we draw some conclusions and comment on the conditional asymptotic distribution.
\begin{rem}\label{rem:cond_chi2}
The asymptotic conditional distribution only depends on $(\hat \eta_1 - \eta)^2$ in (\ref{eq:asymptoticgen}) and not on the signed version, $\hat \eta_1 - \eta$. In particular, in case of HWE (\ref{eq:asymptotic}) remains true if we condition the LRT statistic on $\hat \eta_1^2$ or on $N c \hat \eta_1^2$, the latter being the exact form of the HWE $\chi^2$-statistic. 
\end{rem}
\begin{rem}
Since the LRT statistic and $\chi^2$ statistic are asymptotically equivalent, (\ref{eq:asymptoticgen}) respectively (\ref{eq:asymptotic}) remain true if we replace the LRT statistic by the $\chi^2$ statistic for $H$, 
\[ \mathcal{X}^2 = \sum_{i=1}^2 \sum_{j=1}^3 \frac{\big(n_{i,j}- N_i \hat \pi_j\big)^2}{N_i \hat \pi_j}.\] 
Subsequently, we always use the $\chi^2$ statistic in our conditional test. 
\end{rem}
\begin{rem}
Generally we perform our conditional test under the assumption of HWE in the control group with the distribution as given in corollary \ref{the:mainresult} for the following reasons.
First, in the misspecified situtation, {\it i.e.} $\eta_1 \ne 0$, we have $N \hat \eta_1 \to \infty$ in probability. Therefore, under hypothesis $H$ of equal genotype distributions Theorem 
\ref{the:mainresultgen} and Corollary \ref{the:mainresult} then show that assuming HWE yields an asymptotically conservative test. Second, the asymptotic distribution in Theorem \ref{the:mainresultgen} depends on $\sqrt{N} \big(\hat \eta_1 - \eta \big)$ which requires prior knowledge of $\eta$. Potentially, in some situtations such as population stratification $\eta$ could be estimated under additional assumptions. However, we do not pursue this further in the present paper.
\end{rem}
\begin{rem}
For computing the asymptotic (conditional) P-value of the test statistics in (\ref{eq:asymptotic}), we have to evaluate the distribution function of the asymptotic conditional distribution. This is somewhat technically involved and we refer to the supplementary appendix Section \ref{asec:comput} for details.
\end{rem}
\subsection{Power}
In this section we derive the local power function of our conditional test and compare it to that of the ordinary likelihood ration (or $\chi^2$-) test for homogeneity. 

\begin{theorem}\label{the:localalt}
Suppose that $\eta_1 = 0$ (HWE in controls) and that $\rho_1 = \rho_0$ for some fixed $\rho_0$. Let $\rho_{2,N} = \rho_0 + \rho_\Delta / \sqrt{N}$ and $\eta_{2,N} = \eta_\Delta / \sqrt{N}$, and suppose that for sample size $N$, genotypes in the case group are multinomially distributed with parameters $\rho_{2,N}$ and $\eta_{2,N}$. Let $P_N$ denote the distribution of the observed genotype counts $X_1$ and $X_2$ with the parameters $(\rho_0,0, \rho_{2,N}, \eta_{2,N})$.  \par
If $c=N_1/N$ is bounded away from $0$ and $1$, then for the asymptotic conditional distribution of the LRT statistic, given $\hat \eta_1$, we have as $N \to \infty$ that 
\begin{equation}\label{eq:asymptoticAlt}
	d\Big(\mathcal{D}_{P_N}\Big(
		2 \big(\mathcal{L}(\hat \pi_{1,1},\hat \pi_{1,2},\hat \pi_{2,1},\hat \pi_{2,2})
			- \mathcal{L}(\hat \pi_{1},\hat \pi_{2}) \big)| \hat \eta_1
		\Big), \mathcal{D}\big((Z_1 + Z_2) | \hat \eta_1\Big)\Big) = o_{P_N}(1),
\end{equation}
	where $\mathcal{D}_{P_N}\Big(\cdot| \hat \eta_1\Big)$ is the conditional distribution given $\hat \eta_1$ under $P_N$, the random variable $Z_1$ and the random vector $(Z_2,\hat \eta_1)$ are independent, $Z_1 \sim \chi^2_1\big(\frac{2\, c\,(1-c)\,\rho_\Delta^2}{\rho_0 (1-\rho_0)} \big)$ and $Z_2|\hat \eta_1 \sim c \chi^2_1\big((1-c)\, (\sqrt{N} \hat \eta_1 - \eta_\Delta)^2 \big)$. 
\end{theorem}
In contrast, the ordinary likelihood ratio statistic or $\chi^2$-statistic is asymptotically distributed under the local alternatives of Theorem \ref{the:localalt} as 
\[ 2 \big(\mathcal{L}(\hat \pi_{1,1},\hat \pi_{1,2},\hat \pi_{2,1},\hat \pi_{2,2})
			- \mathcal{L}(\hat \pi_{1},\hat \pi_{2}) \big) \stackrel{P_N}{\to} \chi^2_2\Big( c\,(1-c)\,\big(\frac{2\, \rho_\Delta^2}{\rho_0 (1-\rho_0)} + \eta_\Delta^2\big)\Big).\]
In Figures ?? we compare local asymptotic power curves for both tests for distinct values of $c$, where we take $ \sqrt{N} \hat \eta_1 \sim N(0,1)$ in the conditional case.

\subsection{Covariates}\label{sec:covariates}

In many studies additional covariates are observed and adjusting genetic effects for these covariates is of interest. Here we propose to stratify the sample according to combinations of ordinal covariates and to use a stratified likelihood for estimation and inference. Continuous covariates have to be categorized prior to being used in such an analysis. We assume that covariate combinations are labeled as $k = 1, \ldots, K$, follow a multinomial distribution with parameters $\zeta = (\zeta_1, \ldots, \zeta_K)$ and counts are given by $N_i = N_{i,1} + \ldots + N_{i,K}$ which are distributed across genotypes $j = 1,2,3$ as $N_{i,k} = \sum_{j=1}^3 n_{i,j,k}$. Total data is given by $\mathbf{n} = (\mathbf{n}_1, \ldots, \mathbf{n}_K)$ with $\mathbf{n}_k = (n_{1,1,k}, \ldots, n_{2,3,k} )$. We choose corresponding parameters $\pi_{i,j,k}$, $\pi = (\pi_{1,1,1}, \ldots, \pi_{2,2,K})$ and get the covariate-stratified likelihood $\mathcal{L}_c$ as
\begin{eqnarray*}
	\mathcal{L}_c(\mathbf{n}; \pi, \zeta)
		& = & \sum_{k=1}^K \zeta_k \left[
			\sum_{i=1,2} \sum_{j=1,2} n_{i,j,k} \log \pi_{i,j,k} + \sum_{i=1,2} n_{i,3,k} \log \pi_{i,3,k}
			\right]\\
		& = & \sum_{k=1}^K \zeta_k \mathcal{L}( \mathbf{n}_k ; \pi_{1,1,k},\pi_{1,2,k}, \pi_{2,1, k},\pi_{2,2, k}),
\end{eqnarray*}
where $ \mathcal{L}( \mathbf{n}_k ; \pi_{1,1},\pi_{1,2}, \pi_{2,1},\pi_{2,2})$ denotes the unstratified likelihood for counts $\mathbf{n}_k$. Consider the null hypothesis 
\[ H: \pi_{1,j,k} = \pi_{2,j,k} =: \pi_{j,k} \ \  j = 1, 2, 3; k = 1, ..., K.\]
Set $\pi_0 = (\pi_{1,1}, \ldots \pi_{3,K})$, $(\rho_{i, k}, \eta_{i, k}) = T(\pi_{i, 1, k}, \pi_{i,2, k})$, and  $\eta_1 = (\eta_{1, 1}, \ldots, \eta_{1, k})$, and let $\hat \pi$, $\hat \pi_0$, $\hat \eta_1$ and $\hat \zeta$ be the corresponding vectors of MLEs. 
It follows from Corollary \ref{the:mainresult} that under HWE, the LRT for  is distributed as the following mixture
\begin{eqnarray*}
	d\left(
		\mathcal{D}\Big(
		2 (\mathcal{L}_c(\hat \pi, \hat{\zeta})
			- \mathcal{L}_c(\hat \pi_0, \hat{\zeta}
		)) | \hat \eta_1 \Big),
		\mathcal{D}\Big(\sum_{k=1}^K \hat{\zeta_k} \big(Z_{1,k} + Z_{2,k}(\hat{\eta}_{1, k})\big) \Big)
	\right) = o_P(1),
\end{eqnarray*}

where $Z_{1,k}$ and $\big(Z_{2,k},\hat{\eta}_{1, k}\big)$, $k=1, \ldots, K$, are all independent, $Z_{1,k} \sim \chi^2_1$, $Z_{2,k}(\hat{\eta}_{1, k}) = Z_{2,k}|\hat{\eta}_{1, k}  \sim c_k \chi^2_1\big(N_k \hat \eta_1^2 (1-c_k)\big)$, with $N_k = N_{1,k} + N_{2,k}$, and $c_k = N_{1,k}/N_k$. 

\section{Simulation study}\label{sec_simulation}

In this section, the finite sample properties of the proposed testing procedure are scrutinized. The software package \emph{R} version 2.15.1 was used in all computations \ct{r_development_team_2008}.

\subsection{Single locus simulations}

Table \ref{sim_single} shows simulations comparing the proposed test statistic with the Pearson test and the EHWE test \ct{chen_2007}, which  compares genotype frequencies in cases with an HWE-expectation based on an allele frequency estimation in controls, thereby ignoring the parameter $\eta_1$ entirely. In the $(\rho, \eta)$ parametrization the test is a likelihood ratio (LR) test with null hypothesis 
\[ \Theta_0: \rho_1 = \rho_2, \eta_1 = \eta_2 = 0 \qquad \text{and alternative hypothesis}\qquad  \Theta_1: \eta_1 = 0, (\rho_1 \ne \rho_2 \lor \eta_2 \ne 0),\]
where $(\rho_1, \eta_1)$ are the parameters for controls and $(\rho_2, \eta_2)$ for cases.\par
Table \ref{sim_single} shows in the first section that all tests maintain the $\alpha$ level faithfully under the null hypothesis $H$ with correct specification $\eta_1=0$, as is expected for the sample size of $10^3$ for each group. The second section shows results under the alternative hypothesis, i.e. $H$ does not hold, but with correct specification $\eta_1=0$. 
In cases where $\eta_2 = 0$ as well and the $\rho$s differ, our conditional test is the most powerful test.
The most powerful test in cases where $\rho_1 = \rho_2$ and $0 = \eta_1 \ne \eta_2$ is the EHWE test. 
Generally, the Pearson test can be considered as the least powerful test. This result is expected as the Pearson test makes the least assumptions of the compared tests.\par
Our conditional test and the EHWE test both assume HWE in controls and it appears that they have mutual strengths. However, if we misspecify our test scenario by setting $\eta_1 \ne 0$ some interesting results are apparent from the third and forth sections in Table \ref{sim_single}. In the third section for situations where $\rho_1 = \rho_2, \eta_1 = \eta_2$ (thus $H$ holds true) but $\eta_1 \ne 0$ the conditional test is conservative and does not exhaust the $\alpha$-level of $0.05$ whereas EHWE is anticonservative and exceeds the $\alpha$-level by large margins.\par
In the GWAS situtation this implies that the EHWE test can assign undesirably small P-values to SNPs for which controls deviate from HWE, thus inappropriately increasing their importance, while our conditional test assigns larger P-values thus decreasing their importance.

\subsection{GWASs simulations}\label{sec:gwassim}

We conducted simulations of plausible GWAS scenarios in order to assess the impact of our conditional test on the analysis of such studies. Our single locus simulations indicate that our test should have favorable properties as compared to the EHWE and the Pearson test in a GWAS scenario.

We constructed an alternative based on a logistic penetrance function. If $K$ loci influence disease status, we have:
\begin{eqnarray}\label{sim_penetrance}
	P(Y = 1 | g_1, ..., g_K )
	& = &
		\left(
			1 + \exp\left\{-\left(\mu + \sum_{i = 1}^{K} \beta_i x_i \right)\right\}
		\right)^{-1}
		=: Y(\mathbf{g}; \beta).
\end{eqnarray}

Here, $x_i$ is a score assigned to a genotype, $\beta_i$ is the effect size of locus $i$ and $\mu$ is a baseline penetrance. For genotypes 0 and 2 we assign the respective scores 0 and 1 and for genotype 1 we assign $0, \frac{1}{2}, 1$ to respectively simulate a recessive, additive and dominant mode of inheritance. The prevalence $P(Y = 1)$ of the disease is then given by:

\begin{eqnarray}\label{sim_prevalence}
	\alpha
	& = & P(Y = 1) = \sum_{(g_1,.., g_K) \in \mathbf{G}}
		Y(\mathbf{g}; \beta) P(G = \mathbf{g})
	 = 	\sum_{(g_1,.., g_K) \in \mathbf{G}} Y(\mathbf{g}; \beta) \prod_{i=1}^{K} \rho_{i,g_i},
\end{eqnarray}

assuming that $\mathbf{g} = (g_1,.., g_K), \beta = (\beta_1, ..., \beta_K)$ and $\rho_{i,g_i}$ denotes the genotype frequency of genotype $g_i$ at locus $i$. In all simulations we assume all SNPs to be independent for a random sample from the population.
In order to pick an alternative hypothesis, we deterministically chose some of the parameters and drew the others randomly. Deterministic parameters were $\alpha, \mu, \beta_1, \rho_1$ and the score assignment. $K$, $\beta_2, ..., \beta_K, \rho_2, ..., \rho_K$ were drawn randomly as follows. $\beta_i \sim U(1.1, \beta_1), i = 2, ...$, $\rho_i \sim U(\rho_1/2, \rho_1)$ subject to the constraint $P(Y = 1) = \alpha$.\par
We give further details about the procedure in the supplementary material Section \ref{asec:gwassim}. In summary, we created alternatives for which we know prevalence, maximal effect size and baseline penetrance with freedom in how the prevalance was distributed over effects of additional loci, mimicking the suspected polygenic nature of most complex genetic diseases. From the allele frequencies expected genotype frequencies under HWE were computed and genotypes were drawn according to these distributions for cases and controls for SNPs under the null. For SNPs under the alternative, genotypes for controls were drawn likewise and 
genotypes for cases were drawn according to the penetrance model above keeping only individuals for which phenotype status was realized as diseased.

\subsection{Results of GWASs simulation}

Table \ref{sim_gwas1} summarizes results of the GWASs simulations which compare our conditional test (indicated by $^*$s), EHWE (indicated by $^\bullet$s) and the Person test (plain notation). We chose several measures to judge the performance of the tests. As ranking is important in practical analysis we show rank criteria in columns $Q_x$. We pick the SNP at quantile $x$ according to P-value among SNPs under the alternative and average the rank statistic among all SNPs of these selected SNPs. We also evaluated the number of rejected SNPs denoted by $m_\alpha$ for a false discovery rate (FDR) criterion at level $\alpha$ \ct{benjamini_1995}. Finally, $\Phi_\alpha$ is the power of the FDR procedure, denoting the probability of rejecting any SNP in a GWAS. For all simulations we assumed the prevalence to be $0.1$, sample size to be $2 \times 10^3$ in both groups and averaged all numbers across $10^3$ repetitions after picking an alternative.\par

With respect to ranking (columns $Q_x$), for the additive and the dominant models the conditional test assigns SNPs simulated under the alternative lowest (and hence most favorable) ranks, whereas under the recessive  model EHWE gives lowest ranks. Similarely, the conditional procedure had best power (columns $\Phi_\alpha$) for the dominant and additive models as compared to the other tests and rejected the biggest number of SNPs in these situations. In the recessive case, EHWE had best power as well as most rejected SNPs.

\section{Data Analysis}\label{sec_data}

We applied our testing procedure to a data set previously published \ct{hillmer_2008} and compared it with the same tests as used in the simulation section. The data set is comprised of 296 males with androgenic alopecia and 383 random controls. Details about the data set are given in the initial publication \ct{hillmer_2008}. The data set was pruned for SNPs with a minor allele frequency (MAF) $< 0.1$ ({\it i.e.} $\min\{\rho, 1 - \rho\} < 0.1$) and call rates for SNPs $< 0.9$. 

Quantile-Quantile (Q-Q) plots were used to assess assumptions on underlying distributions (supplementary appendix, Section \ref{asec:dataanal}, Figure 1). The plot for the HWE test in controls (part (A)) shows notable deviation from the expected uniform distribution for expected p-values $< 10^{-2}$. Apart from this deviation in the tail, the HWE statistic shows a close fit with expected values. This type of deviation most likely excludes systematic error such as population stratification (different mixtures of cases and controls as drawn from several underlying unique distributions) or cryptic relatedness (different correlation structure of samples across groups) \ct{devlin_1999,voight_2005} as in principle all SNPs should be affected. To underpin this observation quantitatively, we computed the variance inflation factor described in \ctt{devlin_1999} which was 0.99 and 1.01 for the median and mean based estimates, respectively, in accordance with the expectation of no confounding. A first thing to note is that an arbitrary cut-off for HWE p-
values, say $10^{-3}$ would not remove all deviating SNPs in the dataset. As our test should not be sensitive to non-systematic deviations from HWE, we did not remove any SNPs from further analysis. It is informative to see how the different tests behave in this situation. We subjected the data to our conditional test, the Pearson and the EHWE tests.\\
Part (B), (C) and (D) of this figure show Q-Q-plots for the conditional, Pearson and EHWE tests, respectively. All Q-Q-plots show an excess of low p-values as compared to the expected uniform distribution under the null hypothesis. The EHWE test shows excessive anti-conservative behavior whereas the conditional and Pearson tests both show close correspondence to a uniform distribution for p-values $> 10^{-3.5}$. The data set contains a strong signal on chromosome 20 such that some tail deviation is expected.

Table \ref{data_tabC} shows 30 SNPs ranked according to p-values of the conditional test. Overlap with EHWE is weak except for the top-6 SNPs but the Pearson test shows some agreement. Skipping the first six SNPs, the top ranked SNP according to the conditional test has Rank $78$ for the Pearson test. From the top 10 list according to Pearson rank 5, 6 and 9 are missing from table \ref{data_tabC}. These observations can be described more quantitatively by estimating the correlation between p-values of the test statistics (supplementary appendix, Section \ref{asec:dataanal}, Table 2). Pearson and the conditional test show a correlation of $0.95$ whereas their correlation with EHWE is $< 0.76$. The conditional test is the only test showing a negative correlation with HWE ($-0.01$), EHWE has correlation $\sim 0$ with HWE (as it ignores the HWE statistic in controls) and Pearson has a correlation of $0.24$ with HWE. The negative correlation is not significant in this data set although it is intuitively plausible to assume that HWE should 
have a strictly negative correlation with the conditional test. 

For EHWE the corresponding p-values for HWE in a top list are located in the range $(10^{-9}, 10^{-1})$ with a majority of SNPs having p-values $< 10^{-2}$ (supplementary appendix, Section \ref{asec:dataanal}, Table 1). This reflects the sensitivity of EHWE to misspecification as was shown in the simulation section. It is important to note that EHWE has practically no overlap with either the conditional or Pearson test as determined by SNP rankings. The top-30 list for EHWE contains HWE outliers with p-values $< 10^{-6}$ (ranks 2, 5, 9 in the top-10) whereas the conditional and Pearson tests rank them at ranks $> 200,000$ and $> 3,500$, respectively.

Comparing Pearson and the conditional test, both agreed on the best SNP. The corresponding P-value was $1.2 \times 10^{-7}$ for the conditional test compared with the higher P-value of $4.2 \times 10^{-7}$ for the Pearson test.
eq:fisherinfotranspar

\section{Discussion}\label{sec_discussion}

In this paper we have proposed a new testing procedure for case-control studies in a GWAS setting. One major advantage of our test is that is removes the need for two separate tests (HWE in controls and a test for homogeneity) and thereby eliminates the need for an arbitrary cutoff that is normally introduced to avoid false positives. In order to get a balanced view with respect to the impact of HWE it is necessary to additionally consider the P-values of the HWE goodness-of-fit test as well as the ordinary Pearson test. Also the distribution of the P-values of the HWE goodness-of-fit test should be evaluated by means of, for example, Q-Q plots or the inflation factor of the study in order to detect strong and systematic deviations.

An interesting property of our test is that it is more powerful than competing tests in many situations, in particular in the analysis of GWASs. Intuitively, this result is expected from the distribution of our test statistic. The non-centrality parameter relates linearly to the HWE statistic and therefore gives more mass to the tail for a SNP with a large HWE statistic as compared to an SNP that shows close correspondence with HWE. Therefore, the test has lower critical values for SNPs in HWE as compared to SNPs which deviate from HWE. Loosely speaking, our conditional test is more likely to reject ''promising'' SNPs in HWE and less likely to reject ''atypical'' SNPs with departures from HWE.\par

A more direct approach to address the assumption of HWE in controls, the EHWE test, is generally not effective in the analysis of GWASs. Instead of modeling HWE, EHWE completely ignores the HWE statistic in controls by constraining the parameter $\eta_1$ to zero. While doing so may improve power in single locus simulations under correct specification (i.e. $\eta_1 =0$), it typically decreases power in the GWAS setting and is additionally very sensitive to model misspecification. 


We have analyzed a data set that shows some deviation from the expected distribution for the HWE statistic in controls, that affects about 1\% of all SNPs and could not be attributed to confounding according to the variance inflation estimator of \ctt{devlin_1999}. We chose to re-analyze this data set in its published form as our procedure should perform well even in the presence of non-systematic HWE deviations, {\it i.e.} deviations not attributable to substructure or cryptic relatedness. As the data set contains a very strong signal, it was not surprising to find overlap for the most strongly associated SNPs for all tests that were compared. Except for these SNPs our conditional test and the Pearson test show moderate agreement but no agreement is seen between EHWE and the two other tests. In practical situations the re-ranking of SNPs could have strong consequences as typically the top 30 to 300 SNPs are re-analyzed in independent samples.\par

In conclusion, our test offers a unified and more powerful approach to association testing as compared to standard procedures. These two aspects seem to warrant wide adoption and we aim to integrate our method into standard software packages.\par

\section{Acknowledgments}

We are grateful to Ruth Pfeiffer, Constantin Strauch and Helmut Schäfer for a critical appraisal of the manuscript. Hajo Holzmann gratefully acknowledges financial support from the DFG, grant HO 3260/3-1, the Claussen-Simon Stiftung and from the Landesstiftung Baden-W\"urttemberg (``Juniorprofessorenprogramm''). The alopecia data set was kindly provided by Felix Brockschmidt, Michael Steffens and Markus N\"othen.  


\bibliographystyle{plainnat}

\bibliography{literature}

{\singlespacing

\begin{longtable}{|r|rr|rr|rrr|}
\caption{Single locus simulations comparing the level and the power for the conditional test ($ \Phi^* $), the EHWE test ($ \Phi^\bullet $) and the Pearson $\chi^2$-test ($\Phi$) for testing the hypothesis of homogeneity $H$. Correctly specified models ($\eta_1=0$) and misspecified models $(\eta_1 \not=0)$ were investigated. $10^3$ controls and $10^3$ cases were drawn from the distributions $(\rho_1, \eta_1)$ and $(\rho_2, \eta_2)$, respectively. The significance level was chosen as $\alpha = .05$ and the level respectively the power were estimated from $10^4$ simulations.}\label{sim_single}\\
\hline
& $\rho_1$ & $\eta_1$ & $\rho_2$ & $\eta_2$ & $\Phi_{.05}^*$ & $\Phi_{.05}^\bullet$ & $\Phi_{.05}$ \\
\hline \hline
Level  & 0.10 & 0.00 & 0.10 & 0.00 & 0.048 & 0.051 & 0.049 \\
& 0.50 & 0.00 & 0.50 & 0.00 & 0.053 & 0.053 & 0.053 \\
\hline \hline
Power & 0.10 & 0.00 & 0.10 & 0.05 & 0.199 & 0.320 & 0.173 \\
& 0.50 & 0.00 & 0.50 & 0.05 & 0.179 & 0.273 & 0.155 \\
& 0.10 & 0.00 & 0.10 & 0.10 & 0.772 & 0.990 & 0.791 \\
& 0.50 & 0.00 & 0.50 & 0.10 & 0.575 & 0.820 & 0.507 \\
& 0.10 & 0.00 & 0.11 & 0.00 & 0.149 & 0.138 & 0.135 \\
& 0.50 & 0.00 & 0.51 & 0.00 & 0.088 & 0.084 & 0.082 \\
& 0.10 & 0.00 & 0.12 & 0.00 & 0.457 & 0.426 & 0.423 \\
& 0.50 & 0.00 & 0.52 & 0.00 & 0.209 & 0.188 & 0.188 \\
\hline \hline
Misspec.& 0.10 & 0.05 & 0.10 & 0.05 & 0.029 & 0.313 & 0.049 \\
Level & 0.50 & 0.05 & 0.50 & 0.05 & 0.037 & 0.275 & 0.050 \\
& 0.10 & -0.05 & 0.10 & -0.05 & 0.043 & 0.255 & 0.049 \\
& 0.50 & -0.05 & 0.50 & -0.05 & 0.041 & 0.273 & 0.051 \\
\hline \hline
Misspec. & 0.10 & 0.05 & 0.10 & 0.10 & 0.195 & 0.990 & 0.332 \\
Power & 0.50 & 0.05 & 0.50 & 0.10 & 0.162 & 0.819 & 0.154 \\
& 0.10 & -0.05 & 0.10 & 0.00 & 0.048 & 0.053 & 0.136 \\
& 0.50 & -0.05 & 0.50 & 0.00 & 0.052 & 0.052 & 0.154 \\
& 0.10 & 0.05 & 0.12 & 0.05 & 0.277 & 0.653 & 0.444 \\
& 0.50 & 0.05 & 0.52 & 0.05 & 0.141 & 0.425 & 0.195 \\
& 0.10 & -0.05 & 0.14 & -0.05 & 0.884 & 0.967 & 0.940 \\
& 0.50 & -0.05 & 0.54 & -0.05 & 0.469 & 0.761 & 0.590\\
\hline
\end{longtable}

\begin{sidewaystable}
\begin{longtable}{rrrrr|rrr|rrr|rrr|rrr}
\caption{Simulations of GWAs under different scenarios based on $10^3$ simulations. $OR_1 = exp(\beta_1)$ is the maximal effect size of all loci, $\rho_1$ is the corresponding allele frequency, $\mu$ is the baseline penetrance and $M$ denotes the mode of inheritance and $N$ is the number of SNPs under the alternative. Sample size was $2 \times 10^3$ for controls and cases. $Q_x$ is the average rank of the locus under the alternative with the $x$-quantile according to p-value among SNPs under the alternative. $m_\alpha$ is the average number of rejected SNPs according to an FDR criterion at level $\alpha$ and $\Phi_\alpha$ is the corresponding power (for details see text). A $^*$ represents results for the conditional test, $^\bullet$ represents EHWE and a plain notation indicates the Pearons test.}\label{sim_gwas1}\\
$OR_0$ & $\rho_1$ & $\mu$ & M & $N$ & $Q_0^*$ & $Q_0^\bullet$ & $Q_0$ & $Q_{.25}^*$ & $Q_{.25}^\bullet$ & $Q_{.25}$ & $m_{.05}^*$ & $m_{.05}^\bullet$ & $m_{.05}$ & $\Phi_{.05}^*$ & $\Phi_{.05}^\bullet$ & $\Phi_{.05}$ \\
\hline
1.35 & 0.10 & 0.050 & D & 19 & 12.2 & 20.7 & 28.0 & 673.5 & 907.4 & 1094.3 & 0.1 & 0.1 & 0.1 & 0.132 & 0.109 & 0.075 \\
1.35 & 0.10 & 0.010 & D & 55 & 2.0 & 3.2 & 4.1 & 494.6 & 677.8 & 829.7 & 0.4 & 0.3 & 0.2 & 0.261 & 0.206 & 0.138 \\
1.62 & 0.15 & 0.050 & D & 7 & 1.0 & 1.0 & 1.0 & 2.6 & 2.6 & 2.7 & 4.1 & 4.2 & 3.6 & 0.997 & 0.995 & 0.993 \\
1.62 & 0.15 & 0.010 & D & 25 & 1.0 & 1.0 & 1.0 & 7.1 & 7.2 & 7.5 & 9.6 & 9.4 & 7.6 & 1.000 & 1.000 & 0.997 \\
1.35 & 0.20 & 0.050 & D & 10 & 4.2 & 5.1 & 9.8 & 113.4 & 123.0 & 209.2 & 0.4 & 0.5 & 0.3 & 0.335 & 0.333 & 0.230 \\
1.35 & 0.20 & 0.010 & D & 30 & 1.1 & 1.3 & 1.5 & 50.7 & 58.3 & 100.2 & 1.3 & 1.3 & 0.8 & 0.637 & 0.596 & 0.478 \\
1.62 & 0.20 & 0.050 & D & 5 & 1.0 & 1.0 & 1.0 & 2.0 & 2.0 & 2.1 & 3.6 & 3.7 & 3.2 & 0.999 & 0.999 & 0.990 \\
1.62 & 0.20 & 0.010 & D & 18 & 1.0 & 1.0 & 1.0 & 5.2 & 5.3 & 5.3 & 11.8 & 12.2 & 10.5 & 1.000 & 1.000 & 1.000 \\ \hline
1.50 & 0.10 & 0.050 & A & 22 & 94.0 & 159.3 & 168.4 & 3192.3 & 4448.9 & 4468.2 & 0.0 & 0.0 & 0.0 & 0.021 & 0.016 & 0.010 \\
1.50 & 0.10 & 0.010 & A & 69 & 12.6 & 24.7 & 24.6 & 2839.1 & 4098.8 & 4141.8 & 0.1 & 0.0 & 0.0 & 0.050 & 0.029 & 0.022 \\
1.80 & 0.15 & 0.050 & A & 10 & 2.1 & 4.1 & 4.0 & 47.5 & 69.3 & 75.8 & 0.7 & 0.5 & 0.5 & 0.462 & 0.368 & 0.353 \\
1.80 & 0.15 & 0.010 & A & 32 & 1.0 & 1.2 & 1.2 & 31.9 & 55.9 & 58.3 & 2.0 & 1.3 & 1.3 & 0.772 & 0.625 & 0.616 \\
1.50 & 0.20 & 0.050 & A & 11 & 36.1 & 62.9 & 63.8 & 624.1 & 944.4 & 948.4 & 0.1 & 0.1 & 0.1 & 0.103 & 0.075 & 0.064 \\
1.50 & 0.20 & 0.010 & A & 33 & 2.6 & 4.9 & 5.1 & 278.7 & 492.8 & 489.4 & 0.3 & 0.2 & 0.2 & 0.251 & 0.160 & 0.163 \\
1.80 & 0.20 & 0.050 & A & 7 & 1.1 & 1.4 & 1.4 & 8.0 & 15.0 & 14.0 & 1.6 & 1.2 & 1.2 & 0.766 & 0.667 & 0.647 \\
1.80 & 0.20 & 0.010 & A & 24 & 1.0 & 1.0 & 1.0 & 11.0 & 16.9 & 17.3 & 3.5 & 2.5 & 2.5 & 0.929 & 0.841 & 0.840 \\ \hline
1.65 & 0.20 & 0.050 & R & 59 & 13.6 & 1.8 & 30.0 & 2487.7 & 500.3 & 4494.6 & 0.1 & 0.6 & 0.0 & 0.056 & 0.373 & 0.031 \\
1.65 & 0.20 & 0.010 & R & 101 & 2.4 & 1.0 & 5.0 & 1243.9 & 191.1 & 2464.1 & 0.2 & 3.0 & 0.2 & 0.198 & 0.858 & 0.113 \\
1.98 & 0.25 & 0.050 & R & 25 & 1.0 & 1.0 & 1.2 & 24.0 & 7.3 & 63.6 & 2.3 & 9.0 & 1.3 & 0.847 & 1.000 & 0.670 \\
1.98 & 0.25 & 0.010 & R & 85 & 1.0 & 1.0 & 1.0 & 73.9 & 23.9 & 171.8 & 4.7 & 21.4 & 2.4 & 0.943 & 1.000 & 0.787 \\
1.65 & 0.30 & 0.050 & R & 25 & 1.3 & 1.0 & 2.2 & 90.9 & 11.6 & 209.3 & 0.9 & 4.4 & 0.6 & 0.558 & 0.970 & 0.393 \\
1.65 & 0.30 & 0.010 & R & 84 & 1.0 & 1.0 & 1.3 & 216.4 & 39.7 & 475.1 & 1.6 & 9.1 & 0.9 & 0.652 & 0.993 & 0.478 \\
1.98 & 0.30 & 0.050 & R & 19 & 1.0 & 1.0 & 1.0 & 6.9 & 5.5 & 11.6 & 4.5 & 10.9 & 3.1 & 0.986 & 1.000 & 0.930 \\
1.98 & 0.30 & 0.010 & R & 58 & 1.0 & 1.0 & 1.0 & 18.6 & 15.3 & 28.2 & 10.7 & 27.9 & 6.7 & 0.999 & 1.000 & 0.990 \\
\end{longtable}
\end{sidewaystable}

\newpage

\begin{longtable}{lrrrrrrr}
\caption{Analysis of an alopecia data set. Results are ordered according to P-values of the conditional test and the first 30 SNPs are shown. Columns in parentheses show order statistics. {\it C} denotes the conditional test, {\it P} denotes the Pearson test and {\it E} denotes EHWE. $\lambda$ denotes the NCP of the conditional test.}\label{data_tabC}\\
SNP & C & P & (P) & EHWE & (E) & HWE & $\lambda$ \\
\hline
rs1998076 & 1.21e-07 & 7.15e-07 & 1 & 4.23e-07 & 14 & 5.9e-01 & 0.24 \\
rs6075852 & 2.50e-07 & 1.47e-06 & 2 & 9.65e-07 & 18 & 6.6e-01 & 0.17 \\
rs2180439 & 2.53e-07 & 1.54e-06 & 3 & 1.01e-06 & 20 & 6.6e-01 & 0.17 \\
rs201571 & 8.37e-07 & 4.68e-06 & 7 & 2.44e-06 & 23 & 5.8e-01 & 0.26 \\
rs6113491 & 8.75e-07 & 1.78e-06 & 4 & 3.13e-06 & 25 & 1.2e-01 & 2.01 \\
rs6137444 & 2.17e-06 & 1.02e-05 & 8 & 3.63e-06 & 27 & 6.5e-01 & 0.18 \\
rs10992241 & 3.64e-06 & 2.00e-04 & 78 & 1.08e-05 & 45 & 9.3e-01 & 0.01 \\
rs6137473 & 3.83e-06 & 1.60e-05 & 10 & 1.38e-05 & 50 & 4.9e-01 & 0.41 \\
rs6047768 & 7.16e-06 & 2.43e-05 & 15 & 2.12e-05 & 63 & 4.0e-01 & 0.59 \\
rs2207878 & 7.51e-06 & 2.25e-05 & 13 & 2.05e-05 & 60 & 3.2e-01 & 0.84 \\
rs6113424 & 7.87e-06 & 3.06e-05 & 19 & 2.59e-05 & 75 & 5.1e-01 & 0.37 \\
rs201543 & 7.88e-06 & 3.06e-05 & 17 & 2.59e-05 & 73 & 5.1e-01 & 0.37 \\
rs6035995 & 8.04e-06 & 3.04e-05 & 16 & 2.56e-05 & 71 & 4.7e-01 & 0.44 \\
rs2024885 & 8.08e-06 & 3.06e-05 & 18 & 2.59e-05 & 74 & 5.1e-01 & 0.37 \\
rs1884592 & 9.28e-06 & 3.89e-05 & 26 & 3.27e-05 & 85 & 5.9e-01 & 0.25 \\
rs6137476 & 9.49e-06 & 3.89e-05 & 28 & 3.27e-05 & 87 & 5.9e-01 & 0.25 \\
rs1555264 & 9.53e-06 & 3.64e-05 & 25 & 3.14e-05 & 82 & 5.1e-01 & 0.37 \\
rs6047769 & 9.57e-06 & 3.62e-05 & 23 & 3.09e-05 & 81 & 5.0e-01 & 0.38 \\
rs2328683 & 9.82e-06 & 4.08e-05 & 31 & 3.54e-05 & 93 & 6.5e-01 & 0.18 \\
rs6047731 & 9.84e-06 & 3.89e-05 & 27 & 3.27e-05 & 86 & 5.9e-01 & 0.25 \\
rs4805229 & 1.15e-05 & 4.82e-05 & 33 & 3.60e-05 & 94 & 9.8e-01 & 0.00 \\
rs655683 & 1.53e-05 & 2.13e-05 & 11 & 5.93e-06 & 36 & 1.8e-01 & 1.49 \\
rs927059 & 2.05e-05 & 7.34e-05 & 43 & 6.47e-05 & 109 & 5.3e-01 & 0.33 \\
rs6106434 & 2.05e-05 & 8.50e-05 & 49 & 7.08e-05 & 113 & 5.5e-01 & 0.30 \\
rs1009840 & 2.17e-05 & 5.23e-05 & 35 & 4.90e-05 & 100 & 2.2e-01 & 1.22 \\
rs4896028 & 2.26e-05 & 6.82e-05 & 41 & 6.69e-05 & 110 & 3.1e-01 & 0.86 \\
rs500629 & 2.31e-05 & 8.08e-05 & 46 & 8.71e-05 & 124 & 5.6e-01 & 0.29 \\
rs6137547 & 2.32e-05 & 6.40e-05 & 39 & 5.77e-06 & 34 & 5.3e-01 & 0.38 \\
rs9300398 & 2.65e-05 & 5.04e-05 & 34 & 8.83e-05 & 125 & 2.5e-01 & 1.15 \\
rs4771987 & 3.16e-05 & 1.05e-04 & 56 & 9.49e-05 & 135 & 5.6e-01 & 0.29 \\
\hline
\end{longtable}

}

\appendix
\section{Supplementary appendix}
\small 
%
%
\subsection{Equaivalence of the new parametrization in Section \ref{sec:repar}}\label{asec:parnew}
%
%
%
\subsection{Proofs of Theorem \ref{the:mainresultgen} and Corollary \ref{the:mainresult}}\label{app:proof}
\begin{proof}[{\sl Proof of Theorem \ref{the:mainresultgen}}]
We set 
\[\mathcal{\tilde L}(\rho_1,\eta_1,\rho_2,\eta_2) = \mathcal{L}\big(T^{-1} (\rho_1,\eta_1),T^{-1}(\rho_2,\eta_2)\big),\qquad \mathcal{\tilde L}(\rho,\eta) = \mathcal{L}\big(T^{-1} (\rho,\eta)\big).\]
We let $(\hat \rho_1,\hat \eta_1,\hat \rho_2,\hat \eta_2)$ and $(\hat \rho,\hat \eta)$ be the unrestricted and restricted ML estimators of the transformed parameters, respectively. Denote the true values of the parameters by $\rho_1 = \rho_2 = \rho$ and $\eta_1 = \eta_2 = \eta$.  
For ease of notation set
\begin{eqnarray*}
E & = & (1 + \eta),  \qquad  C = \frac{c}{2 (1-c)}, \qquad K = \left(1 - \eta +\frac{(1 - c) E\eta^2 (1 -2 \rho )^2}{V}\right),\\
R & = & \rho (1 - \rho),\qquad  V = \eta^2(1 -2 \rho )^2 -2 R + \eta  (2 - 6 R)
\end{eqnarray*}

\begin{eqnarray}
\FI	& = & \frac{1}{2} \left(
\begin{array}{ll}
  (1 -\eta) R &  \eta  E (1 - 2 \rho) \\
 \eta  E (1 - 2 \rho) & - EV / R
\end{array}
\right)\label{eq:fisherinfotransparAppendix}
\end{eqnarray}
A standard argument in likelihood theory (see {\it e.g.} \ct{ferguson_1996}, p. 145, eq (4)) now shows that under our assumptions, 
\begin{eqnarray}\label{eq:loglikecond}
& & 2 \big( \mathcal{L}(\hat \rho_1, \hat \eta_1,\hat \rho_2,\hat \eta_2)-\mathcal{L}(\hat \rho,\hat \eta) \big)\nonumber\\
 & = & N \big(\hat \rho_1 - \hat \rho, \hat \rho_2-\hat \rho,\hat \eta_1 - \hat \eta, \hat \eta_2-\hat \eta \big) \FIjoor \big(\hat \rho_1 - \hat \rho, \hat \rho_2-\hat \rho,\hat \eta_1 - \hat \eta, \hat \eta_2-\hat \eta \big)^T +o_P(1)
\end{eqnarray}
where $\FI$ is as in (\ref{eq:fisherinfotranspar}) and 
\begin{eqnarray*}
\FIjoi	& = & \frac{1}{2} \left(
\begin{array}{llll}
  \frac{(1 -\eta) R}{c}&  0 &   \frac{\eta  E (1 - 2 \rho)}{c} & 0\\	
 0 &   \frac{(1 -\eta) R}{(1-c)}  &  0 &  \frac{\eta  E (1 - 2 \rho)}{(1-c)} \\
 \frac{\eta  E (1 - 2 \rho)}{c} & 0 & -\frac{EV }{cR}& 0 \\    
0 &   \frac{\eta  E (1 - 2 \rho)}{(1-c)}  &  0 & -\frac{EV }{R (1-c)} \\
\end{array}
\right)
\end{eqnarray*}
Now, straightforward computations give the covariance in the following normal limit,
\[  \sqrt{N} \big((\hat{\pi}_{11}, \hat{\pi}_{12}, \hat{\pi}_{21}, \hat{\pi}_{12}, \hat{\pi}_{1}, \hat{\pi}_{2}) -(\pi_1, \pi_2,\pi_1, \pi_2,\pi_1, \pi_2) \big) \to N(0,\Sigma_\pi), \]
where $(\pi_1, \pi_2) = T^{-1}(\rho,\eta)$ and
\[
\Sigma_{\pi} = 
\left(
\begin{array}{llllll}
 \frac{\left(1-\pi_1\right) \pi_1}{c} & -\frac{\pi_1 \pi_2}{c} & 0 & 0 & \left(1-\pi_1\right) \pi_1 & -\pi_1 \pi_2 \\
 -\frac{\pi_1 \pi_2}{c} & \frac{\left(1-\pi_2\right) \pi_2}{c} & 0 & 0 & -\pi_1 \pi_2 & \left(1-\pi_2\right) \pi_2\\
 0 & 0 & \frac{\left(1-\pi_1\right) \pi_1}{1-c} & -\frac{\pi_1 \pi_2}{1-c} & \left(1-\pi_1\right) \pi_1 & -\pi_1 \pi_2 \\
 0 & 0 & -\frac{\pi_1 \pi_2}{1-c} & \frac{\left(1-\pi_2\right) \pi_2}{1-c} & -\pi_1 \pi_2 & \left(1-\pi_2\right) \pi_2 \\
 \left(1-\pi_1\right) \pi_1& -\pi_1 \pi_2 & \left(1-\pi_1\right) \pi_1 & -\pi_1 \pi_2 & \left(1-\pi_1\right) \pi_1 &
   -\pi_1 \pi_2\\
 -\pi_1 \pi_2& \left(1-\pi_2\right) \pi_2 & -\pi_1 \pi_2& \left(1-\pi_2\right) \pi_2& -\pi_1 \pi_2 &
   \left(1-\pi_2\right) \pi_2
\end{array}
\right)
\]
Therefore, from the $\delta$-method, 
\begin{equation*}
\sqrt{N} \left( \begin{array}{c} \hat \rho_1 - \rho \\ \hat \eta_1 - \eta \\ \hat \rho_2 - \rho \\ \hat \eta_2 - \eta \\ \hat \rho - \rho \\ \hat \eta - \eta \end{array} \right) \to N(0,\FIcom), \qquad 
\FIcom = 
\left(
\begin{array}{lll}
  \FI / c &  0 & \FI\\
 0 &  \FI / (1-c) & \FI\\
 \FI & \FI & \FI
\end{array}
\right)
\end{equation*}
Since
\begin{equation*}
\left(
\begin{array}{l}
\hat \rho_1 - \hat \rho \\
\hat \rho_2 - \hat \rho \\
\hat \eta_1  - \hat \eta \\
\hat \eta_2 - \hat \eta \\
\hat \eta_1 - \eta  
\end{array}
\right) \quad 
 =  A \cdot \quad 
\left(
\begin{array}{l}
\hat \rho_1 - \rho  \\
\hat \eta_1 - \eta  \\
\hat \rho_2 - \rho  \\
\hat \eta_2 - \eta  \\ 
\hat \rho - \rho  \\
\hat \eta - \eta 
\end{array}\right)
, \qquad \text{where} \qquad 
A = \left(
\begin{array}{cccccc}
1 & 0 & 0 & 0& -1 & 0 \\
0 & 0 & 1 & 0 & -1 & 0 \\
0 & 1 & 0 & 0 & 0 & -1\\
0 & 0 & 0 & 1 & 0 & -1 \\
0 & 1 & 0 & 0 & 0 & 0 \\
\end{array}
\right) 
\end{equation*}
it follows that
\begin{equation}\label{eq:asympcondmajor}
 \sqrt{N} \big(
\hat \rho_1 - \hat \rho,  \hat \rho_2-\hat \rho,  \eta  - \hat \eta,  \hat \eta_2-\hat \eta,  \hat \eta_1 - \eta \big)^T  
 \to N (0, \widetilde{\mathcal{J}}), \qquad \text{ where } \ \widetilde{\mathcal{J}} = A \FIcom  A^T.
\end{equation}
Partitioning
\begin{equation}\label{eq:hilfsmatrix}
 \widetilde{\mathcal{J}} = \left(
\begin{array}{cc}
\widetilde{\mathcal{J}}_{1,1}  & \widetilde{\mathcal{J}}_{1,2} \\
\widetilde{\mathcal{J}}_{2,1}  & \widetilde{\mathcal{J}}_{2,2} 
\end{array}
\right), \qquad \widetilde{\mathcal{J}}_{1,1} \in \R^{4 \times 4},\ \widetilde{\mathcal{J}}_{2,2} \in \R^{1 \times 1}\ \widetilde{\mathcal{J}}_{2,1} = \widetilde{\mathcal{J}}_{1,2}^T \in \R^{1 \times 4},  
\end{equation}

the components of $ \widetilde{\mathcal{J}}$ in (\ref{eq:hilfsmatrix}) are evaluated as
%
\begin{eqnarray*}
\widetilde{\mathcal{J}}_{1,1} & = &	\left(
\begin{array}{ccccc}
	- (\eta - 1) R/ (4C)  & \frac{1}{2} (\eta -1)R 
		& E \eta (1 - 2 \rho)/ (4C) & -\frac{1}{2} E\eta (1 - 2 \rho)\\
	\frac{1}{2} (\eta -1)R   & -C (\eta -1) R
	& -\frac{1}{2} E\eta (1 - 2 \rho)  &  C E \eta (1 - 2 \rho)\\
	E \eta (1 - 2 \rho)/ (4C) & -\frac{1}{2} E\eta (1 - 2 \rho)
	& -EV / (4CR) &  EV / (2R)\\
	-\frac{1}{2} E \eta (1 - 2 \rho) &  C E \eta (1 - 2 \rho)
	& EV / (2R) & -CEV / R
\end{array}
\right)
\end{eqnarray*}

\begin{eqnarray*}
\widetilde{\mathcal{J}}_{2,1} & = &	\left(
	 E \eta (1 - 2 \rho)/(4C),
	-E\eta (1 - 2 \rho)/2,
	-EV / (4CR),
	-EV / (2R)
	\right),\\
\widetilde{\mathcal{J}}_{2,2} & = &	 -EV / (2cR)
.
\end{eqnarray*}

From (\ref{eq:asympcondmajor}), and the formulas for $ \widetilde{\mathcal{J}}$, we have for the asymptotic conditional distribution 
\begin{equation}\label{eq:asympcond}
 d\Big(\mathcal{D}\Big(\sqrt{N} \big(\hat \rho_1 - \hat \rho, \hat \rho_2-\hat \rho, \hat \eta_1 - \hat \eta, \hat \eta_2-\hat \eta \big)^T| \hat \eta_1\Big), N(\mu(\hat \eta_1), \Sigma)\Big) \to 0 \qquad \text{ in probability},
\end{equation}
where 
\[
	\mu(\hat \eta_1) = \sqrt{N} (\hat \eta_1- \eta )\, 
		(\widetilde{\mathcal{J}}_{1,2} \widetilde{\mathcal{J}}_{2,2}^{-1}), \qquad
		\Sigma = \widetilde{\mathcal{J}}_{1,1}
		- \widetilde{\mathcal{J}}_{1,2} \widetilde{\mathcal{J}}_{2,2}^{-1} \widetilde{\mathcal{J}}_{2,1}.
\]
Then $\mu(\hat \eta_1)$ is evaluated as (\ref{eq:noncentralparcond}), and $\Sigma$ as

\begin{eqnarray}\label{eq:covasymcondright}
\Sigma & = &
	\left(
	\begin{array}{cccc}
	R K/ (4C)
	& - R K / 2
	& (1-c) E \eta (1 - 2 \rho)/2
	& -c	E \eta (1 - 2 \rho)/2 \\
	- R K / 2
	& C K R
	& -c	E \eta (1 - 2 \rho)/2
	& c C E \eta (1 - 2 \rho) \\
	 (1-c) E \eta (1 - 2 \rho)/2
	&  -c	E \eta (1 - 2 \rho)/2
	& -(1 - c) E V / (2R)
	& c E V / (2R)\\
	-c	E \eta (1 - 2 \rho)/2 
	&  c C E \eta (1 - 2 \rho)
	& c E V / (2R)
	& - c C E V	/ R 
\end{array}
	\right).
\end{eqnarray}

From (\ref{eq:loglikecond}) we deduce the desired result. 
\end{proof}  
\begin{proof}[{\sl Proof of Corollary \ref{the:mainresult}}]
In case of HWE, i.e.~$\eta =0$, we have
\begin{eqnarray*}
E  =  1,\quad  K = 1,\quad V= -2 R, 
\end{eqnarray*}
\end{proof}
so that $		\mu(\hat \eta_1) = 	\big(	0, 0, 1-c,
	-c\big)\, \sqrt{N} \hat \eta_1,
$ 
\begin{eqnarray*}
\Sigma & = &
	\left(
	\begin{array}{cccc}
	(1-c)R  / 2c
	& -\frac{1}{2} R  
	& 0
	& 0\\
	-\frac{1}{2} R  
	& c R / 2(1-c)
	& 0
	& 0 \\
	0
	& 0
	& 1-c
	& -c \\
0
	&  0
	&  -c
	& \frac{  c^2}{(1-c)} 
	\end{array}
	\right),\\
\FIjoor & = & \text{diag}\, \Big(
  \frac{2\, c}{R}, \frac{2\, (1-c)}{R}  , c, (1-c)\Big) 
\end{eqnarray*}
Therefore, 
\begin{equation*}
 d\Big(\mathcal{D}\Big(\sqrt{N} \big(\hat \rho_1 - \hat \rho, \hat \rho_2-\hat \rho, \hat \eta_1 - \hat \eta, \hat \eta_2-\hat \eta \big)^T| \hat \eta_1\Big), Z\Big) \to 0 \qquad \text{ in probability},
\end{equation*}
where for independent $X, Y \sim N(0,1)$, both independent of $\hat \eta_1$,  
\begin{eqnarray*}
Z & = & c (1-c)/c X^2 + (1-c) c / (1-c) X^2 \\
	&  + &  c \big( (1-c)^{1/2} Y + (1-c) \sqrt{N} \hat \eta_1 \big)^2 + (1-c)\, \Big(\frac{c}{(1-c)^{1/2}}Y + c \sqrt{N} \hat \eta_1 \Big)^2\\
	& = & X^2 + c \Big(Y +\sqrt{N} \hat \eta_1 (1-c)^{1/2} \Big)^2,
\end{eqnarray*}
and the assertion follows. 
%
%
\subsection{Numerical computation of P-values}\label{asec:comput}
The asymptotic conditional distribution under the hypothesis of equal genotype distributions and HWE in Corollary \ref{the:mainresult} is the convolution of the $\chi^2_1$-variable and a rescaled non-central $\chi^2_1$ variable. Now, the distribution function $F_1(x;\lambda)$ of $\chi^2_1(\lambda )$ ($\chi^2_1$ with non-centrality parameter (NCP) $\lambda$) can be written as ({\it cf.} \ctp{johnson_1994}) 
\[ F_1(x; \lambda) = \sum_{j=0}^\infty \Big(\frac{(\lambda/2)^j}{j!} e^{-\lambda/2} \Big) F_{1 + 2j}(x),\]
where $F_\nu(x)$ is the distribution function of the central $\chi^2_\nu$ distribution. Therefore, it would be sufficient to compute the convolutions of $\chi^2_1$ and the rescaled $c \chi^2_{1 + 2j}$, $j \geq 0$. However, this amounts to computing the convolution of two gamma variables with distinct scale parameters ({\it cf.} \ctp{johnson_1994}). Although methods for this problem do exist, we found it preferable to simply compute the convolution of the densities of $\chi^2_1$ and $c \chi^2_1(\lambda)$  numerically, as subsequently outlined.

Numerical integration is used to compute the distribution function of the convolution $Z := Z_1 + c Z_2(\lambda)$, with NCP $\lambda = N \hat{\eta}_1 (1 - c)$. Denote with $\varphi(x, y) = \varphi(x,y; \lambda, c)$ the joint density of $Z_1$ and $c Z_2(\lambda)$. Therefore $F_Z(x) = \int_0^x \int_0^z \varphi(t, z - t) dt dz$, if $F_Z$ denotes the distribution function of $Z$. A first change of variables allows for a rectangular integration area: $\theta: (x, y) \to (x, x y)$,
	$F_Z(x) = \int_0^x \int_0^1 \varphi \circ \theta(z, t) |D \theta| dt dz$.
Note, that $\varphi \circ \theta(z, t)$ has poles at $t = 0$ and $t = 1$ for all values $z$. The poles behave as $\frac{1}{\sqrt{t}}$ for $t \to 0$ and as $\frac{1}{\sqrt{1 - t}}$ for $t \to 1$. For $t$, we split the integral into subintervals $(0, \frac{1}{4})$, $(\frac{1}{4}, \frac{3}{4})$, $(\frac{3}{4}, 1)$ and apply transformations in the outer intervals in order to eliminate the poles. To this end, we use the identities $\int_0^b f(x) = \int_0^{\sqrt{b}} 2t f(t^2)$ for the pole at $0$ and $\int_a^1 f(x) = \int_0^{\sqrt{1-a}} 2t f(1 - t^2)$ for the pole at $1$.

Therefore, after some algebraic transformations, the distribution function of $Z$ is given by:
\begin{eqnarray*}
	&	& F_Z(x;\lambda, c)\\
	& = & \int_0^{\frac{1}{4}} \int_0^x
		\frac{2}{\pi}
		\exp \left\{-\frac{1}{2}\left(\lambda + s + 4 t^2 s (\frac{1}{c} - 1)
			+ \log(c(1 - 4t^2))\right)\right\} 
		\cosh\left(2 t \sqrt{\frac{\lambda s}{c}}\right)ds\ dt\\
	& +	& \int_{\frac{1}{4}}^{\frac{3}{4}} \int_0^x
		\frac{1}{2 \pi}
		\exp \left\{-\frac{1}{2}\left(\lambda + s + s t (\frac{1}{c} - 1)
			+ \log(c t (1 - t))\right)\right\} 
		\cosh\left( \sqrt{\frac{\lambda s t}{c}}\right)ds\ dt\\
	& +	& \int_{\frac{3}{4}}^1 \int_0^x
		\frac{2}{\pi}
		\exp \left\{-\frac{1}{2}\left(\lambda + \frac{s}{c} + 4 t^2 s (1 - \frac{1}{c})
			+ \log(c(1 - 4t^2))\right)\right\} 
		\cosh\left(\sqrt{\frac{\lambda s (1 - 4 t^2)}{c}}\right)ds\ dt,
\end{eqnarray*}

which is amenable to numerical integration. We use the software package \emph{adapt} in \emph{R} version 2.8.1 to compute the integrals.

%
%
\subsection{Proof of Theorem \ref{the:localalt}}
\begin{proof}
Similarly as above, setting $a = \rho_0 (1- \rho_0)/2$ we have under $P_N$ that
\begin{align}\label{eq:loglikecondalt}
&  2 \big( \mathcal{L}(\hat \rho_1, \hat \eta_1,\hat \rho_2,\hat \eta_2)-\mathcal{L}(\hat \rho,\hat \eta) \big)\nonumber\\
	 = & N \Big(\frac{c}{a} \big( \hat \rho_1 - \hat \rho\big)^2 + \frac{1-c}{a} \big( \hat \rho_2 - \hat \rho\big)^2 \Big) + N \Big(c \big( \hat \eta_1 - \hat \eta\big)^2 + (1-c)\, \big( \hat \eta_2 - \hat \eta\big)^2 \Big)+ o_{P_N}(1).
\end{align}

Therefore, from the $\delta$-method, 
\begin{equation}\label{eq:limitjoint}
\sqrt{N} \big((\hat \rho_1,\hat \eta_1,\hat \rho_2,\hat \eta_2, \hat \rho,\hat \eta )^T - (\rho_0, 0,\rho_0,0,\rho_0,0)^T \big) \stackrel{P_N}{\to} N\big(\mu_\Delta,\FIcom\big), 
\end{equation}
where
\[
\FIcom = 
\left(
\begin{array}{llllll}
 a/c & 0 & 0 & 0 & a & 0 \\
0 & b/c & 0 & 0 & 0 & b\\
 0 & 0 & a/(1-c) & 0 & a & 0 \\
 0 & 0 & 0 & b/(1-c) & 0 & b \\
 a& 0 & a & 0 & a & 0\\
0 & b & 0 & b & 0  & b
\end{array}
\right), \qquad \mu_\Delta \left(
\begin{array}{c} 0 \\ 0 \\ \rho_\Delta \\ \eta_\Delta \\ (1-c)\,\rho_\Delta \\ (1-c)\,\eta_\Delta \end{array}
\right)
\]
Thus, the two squares in (\ref{eq:loglikecond}) are asymptotically independent. For the term involving the $\eta$'s, we note that 
\[ d\Big( \mathcal{D}_{P_N}\Big(\sqrt{N} \big(\hat \eta_2,\hat \eta \big)|\hat \eta_1\Big), \big(N(\mu_{\hat \eta_1},\Sigma_{cond})|\hat \eta_1\big)\Big) = o_{P_N}(1),\]
where
\[ \mu_{\hat \eta_1} = (\eta_\Delta,(1-c)\, \eta_\Delta + c \sqrt{N} \hat \eta_1)^T,\qquad \Sigma_{cond} =  \left(\begin{array}{cc}  1/(1-c) & 1 \\ 1 & (1-c)  \end{array} \right).\]
Thus,
\[ d\Big(\mathcal{D}\Big(N \Big(c \big( \hat \eta_1 - \hat \eta\big)^2 + (1-c)\, \big( \hat \eta_2 - \hat \eta\big)^2|\hat \eta_1\Big), \mathcal{D}\Big(W| \hat \eta_1 \Big)\Big) = o_P(1),\]
where 
\[ W = c \big((1-c)^{1/2} X - (1- c) (\sqrt{N} \hat \eta_1 - \eta_\Delta)\big)^2 + (1-c)\, \big((1-c)^{1/2} X + c (\sqrt{N} \hat \eta_1- \eta_\Delta)- X/(1-c)^{1/2}\big)^2 \]
and $X\sim N(0,1)$ is independent of $\hat \eta_1$. $W$ reduces to
 \[ W|\hat \eta_1  = c \Big(X - (1-c)^{1/2}\big(\sqrt{N} \hat \eta_1 - \eta_\Delta\big)  \Big)^2|\hat \eta_1  \sim c \chi^2_1\Big((1-c) \, \big(\sqrt{N} \hat \eta_1 - \eta_\Delta\big)^2\Big).\]
The other term in (\ref{eq:loglikecondalt}) is dealt with similarly, and the theorem follows. 
\end{proof}
%

%
\subsection{Details on GWAS simulation set-up in Section \ref{sec:gwassim}}\label{asec:gwassim}

We here give details about the simulation of genotypes for the GWAS simulations. With the penetrance model
\begin{eqnarray*}\label{sim_penetranceAppendix}
	P(Y = 1 | g_1, ..., g_K )
	& = &
		\left(
			1 + \exp\left\{-\left(\mu + \sum_{i = 1}^{K} \beta_i x_i \right)\right\}
		\right)^{-1}
		=: Y(\mathbf{g}; \beta).
\end{eqnarray*}
and the choosing of parameters as described in the simulation section the alternative hypothesis is fully specified. However, computationally it is not straightforward to choose the number of loci under the alternative $K$ and random penetrance parameters $\mathbf{\beta} = (\beta_2, ..., \beta_K)$ such that
$$
	\alpha = P(Y = 1),
$$

for a given $\alpha$. In order to efficiently draw parameters we use a step-wise procedure. We construct a parameter vector $\theta^{(i)} = (\alpha, \mu, \beta_1, \beta_2, ..., \beta_{i + 1}, \rho_1, \rho_2, ..., \rho_{i + 1})$ by adding parameters $\beta_{i + 1} \sim U(1.1, \beta_1)$, $\rho_{i + 1} \sim U(\rho_1/2, \rho_1)$ to $\theta^{(i - 1)}$. We then estimate prevalence $P(Y = 1)$ by $\hat{\alpha}$ and accept $\theta^{(i)}$ if $\hat{\alpha} < \alpha$. Otherwise we draw new parameters $\beta_{i + 1}, \rho_{i + 1}$. We stop the procedure, if $\hat{\alpha} \in (\alpha - \epsilon, \alpha + \epsilon)$. We use $\epsilon = 10^{-2}$ in all simulations.\par

As computation of $\alpha$ would require summation over all possible genotype combinations in formula (4), the number of which grows exponentially with the number of loci, we instead estimate $\alpha$ by Monte-Carlo integration. Instead, we compute $\hat{\alpha} = \sum_{(g_1,.., g_K) \in \mathbf{G_0}} Y(\mathbf{g}; \beta) G(\mathbf{g}; \rho)$, with $\mathbf{G_0} = (g_{1j},.., g_{ij})_{j=1}^{M}$, where each $g_{lj}$ is independently drawn from the corresponding genotype distribution of locus $l$. We use $M = 10^5$ in the simulations.\par
After specifying the alternative, we independently draw genotypes from the distribution as specified by $\rho = (\rho_1, ..., \rho_K)$ and assign a phenotype according to the penetrance model. As we typically use $\alpha = .1$ in the simulations, the ratio of generated cases and controls is around 1:9, which is a tolerable excess of controls.\par
Finally we simulate $S$ control loci for which the distribution between cases and controls is identical. We use $S = 3 \times 10^5$ in the simulations. We average results from simulations over $10^3$ runs for each alternative that was generated once as outlined above.

\newpage

\subsection{Further tables and figures for data analysis}\label{asec:dataanal}

\begin{longtable}{lrrrrrrr}
\caption{Analysis of an alopecia data set. Results are ordered according to P-values of EHWE and the first 30 SNPs are shown. Columns in parentheses show order statistics. {\it C} denotes the conditional test, {\it P} denotes the Pearson test and {\it E} denotes EHWE. $\lambda$ denotes the NCP of the conditional test.}\label{data_tabE}\\
SNP & EHWE & C & (C) & P & (P) & HWE & $\lambda$ \\
\hline
rs1506694 & 1.13e-11 & 8.17e-03 & 2195 & 3.43e-03 & 900 & 7.9e-03 & 4.25 \\
rs12357377 & 1.13e-10 & 9.68e-01 & 245387 & 2.23e-01 & 57416 & 4.5e-07 & 16.24 \\
rs2189935 & 1.88e-10 & 2.47e-01 & 63594 & 4.02e-02 & 10418 & 2.2e-04 & 6.93 \\
rs3760877 & 1.07e-09 & 1.51e-01 & 39079 & 1.44e-02 & 3768 & 6.2e-03 & 6.79 \\
rs396999 & 3.17e-09 & 1.00e+00 & 253779 & 6.01e-01 & 153007 & 7.2e-09 & 23.04 \\
rs1873921 & 6.30e-09 & 1.06e-01 & 27530 & 2.40e-03 & 640 & 3.1e-02 & 6.12 \\
rs13362504 & 3.45e-08 & 1.73e-02 & 4630 & 4.57e-03 & 1208 & 8.4e-03 & 4.91 \\
rs11232869 & 7.84e-08 & 7.72e-01 & 195223 & 2.85e-01 & 73175 & 3.3e-05 & 7.20 \\
rs1491485 & 1.45e-07 & 9.71e-01 & 246040 & 6.68e-02 & 17202 & 1.7e-09 & 22.81 \\
rs12282 & 1.50e-07 & 2.36e-04 & 80 & 6.35e-04 & 184 & 9.4e-01 & 0.00 \\
rs632547 & 1.63e-07 & 9.88e-01 & 250757 & 4.88e-01 & 124485 & 2.1e-06 & 12.32 \\
rs7857803 & 1.96e-07 & 5.35e-01 & 135381 & 8.97e-02 & 23171 & 3.5e-03 & 7.71 \\
rs4307321 & 3.35e-07 & 8.29e-01 & 209871 & 2.02e-01 & 52160 & 4.7e-04 & 10.17 \\
rs1998076 & 4.23e-07 & 1.21e-07 & 1 & 7.15e-07 & 1 & 5.9e-01 & 0.24 \\
rs647731 & 5.53e-07 & 5.18e-02 & 13690 & 3.01e-02 & 7797 & 6.2e-02 & 2.49 \\
rs2687860 & 8.08e-07 & 1.00e+00 & 253697 & 3.76e-01 & 96292 & 1.5e-07 & 24.61 \\
rs295117 & 9.08e-07 & 9.85e-01 & 249783 & 4.01e-01 & 102635 & 8.9e-05 & 13.95 \\
rs6075852 & 9.65e-07 & 2.50e-07 & 2 & 1.47e-06 & 2 & 6.6e-01 & 0.17 \\
rs9285864 & 9.70e-07 & 8.32e-01 & 210441 & 1.49e-02 & 3907 & 1.2e-07 & 21.28 \\
rs2180439 & 1.01e-06 & 2.53e-07 & 3 & 1.54e-06 & 3 & 6.6e-01 & 0.17 \\
rs2975520 & 1.19e-06 & 1.00e+00 & 253844 & 8.32e-01 & 211079 & 1.0e-09 & 25.35 \\
rs875001 & 1.42e-06 & 1.49e-01 & 38498 & 6.47e-02 & 16671 & 3.4e-02 & 3.19 \\
rs201571 & 2.44e-06 & 8.37e-07 & 4 & 4.68e-06 & 7 & 5.8e-01 & 0.26 \\
rs2425628 & 2.78e-06 & 9.71e-01 & 246163 & 4.55e-01 & 115990 & 2.7e-04 & 10.94 \\
rs6113491 & 3.13e-06 & 8.75e-07 & 5 & 1.78e-06 & 4 & 1.2e-01 & 2.01 \\
rs970952 & 3.51e-06 & 6.56e-01 & 165884 & 2.47e-01 & 63608 & 1.9e-03 & 6.09 \\
rs6137444 & 3.63e-06 & 2.17e-06 & 6 & 1.02e-05 & 8 & 6.5e-01 & 0.18 \\
rs10824842 & 3.93e-06 & 1.00e+00 & 253840 & 3.42e-01 & 87705 & 8.3e-04 & 32.92 \\
rs1555257 & 4.53e-06 & 3.37e-03 & 965 & 3.01e-03 & 803 & 2.5e-01 & 1.03 \\
rs7744253 & 4.67e-06 & 8.75e-01 & 221379 & 3.55e-01 & 90907 & 9.4e-04 & 7.74 \\
\hline
\end{longtable}

\begin{longtable}{r|rrrr}
\caption{Correlation structure between the conditional, EHWE, Pearson, HWE tests for the obesity data set.}\label{data_cor}\\
 & Conditional & EHWE & Pearson & HWE \\
\hline
Conditional & 1.00 & 0.76 & 0.95 & -0.01 \\
EHWE & 0.76 & 1.00 & 0.72 & 0.01 \\
Pearson & 0.95 & 0.72 & 1.00 & 0.24 \\
HWE & -0.01 & 0.01 & 0.24 & 1.00 \\
\hline
\end{longtable}

\renewcommand{\floatpagefraction}{.95}
\begin{figure}[!ht]\label{data_qqPlot}
	\caption{Q-Q-plots for p-values derived from the obesity data set comparing an empirical distribution with a uniform distribution on a double logarithmic scale. Part (A) shows the HWE test in controls, (B) is the conditional test, (C) the Pearson test and (D) is EHWE.}
	\includegraphics[width=.9\hsize]{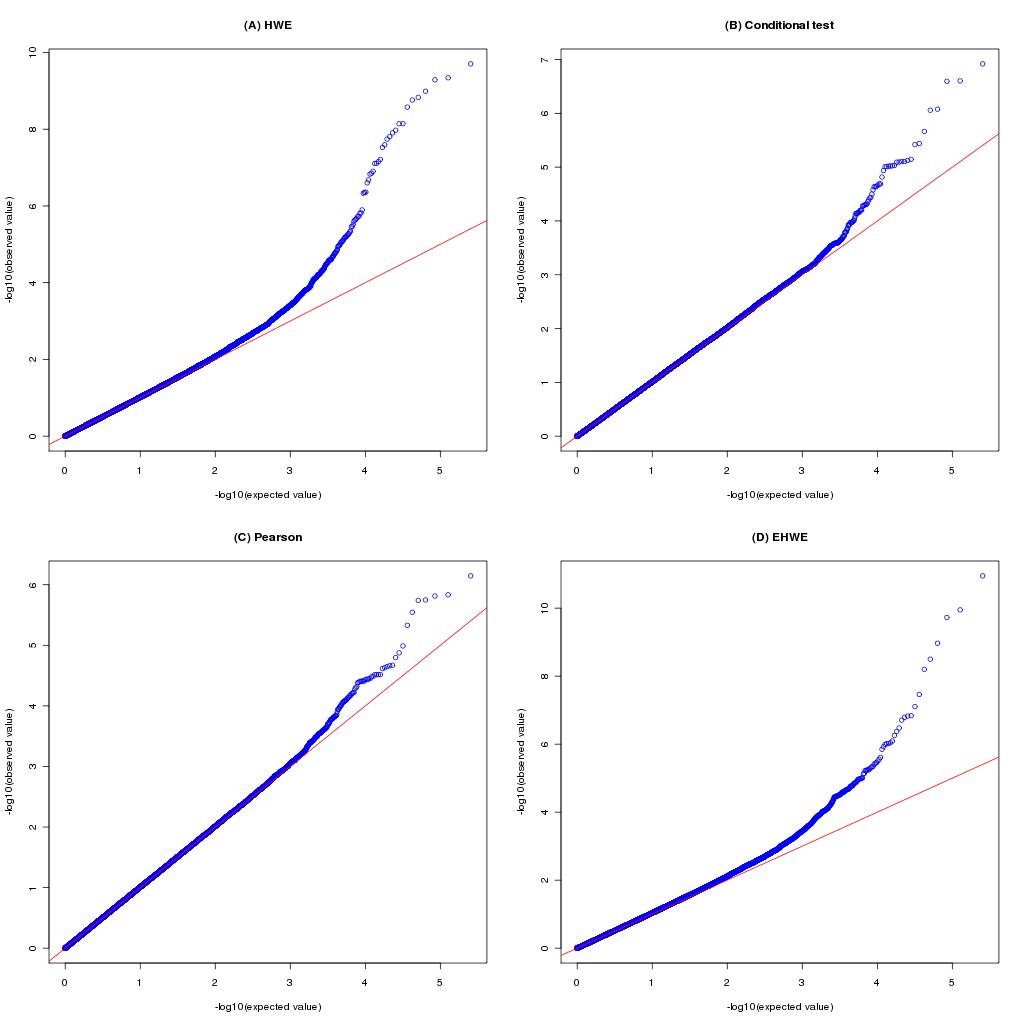}
\end{figure}

\end{document}